\documentclass[letterpaper]{article} 
\usepackage{aaai2026}  
\usepackage{times}  
\usepackage{helvet}  
\usepackage{courier}  
\usepackage[hyphens]{url}  
\usepackage{graphicx} 
\def\UrlFont{\rm}  
\usepackage{natbib}  
\usepackage{caption} 
\usepackage{amsmath}
\usepackage{amssymb}
\usepackage{amsthm}
\usepackage{booktabs}
\usepackage{array}
\usepackage{microtype}
\nocopyright

\newtheorem{definition}{Definition}
\newtheorem{theorem}{Theorem}
\newtheorem{proposition}{Proposition}

\title{Protocol-Driven Development: Governing Generated Software Through Invariants and Continuous Evidence}
\author{
    Jun He,
    Deying Yu
}
\affiliations{
    OpenKedge.io\\
    junhe@openkedge.io, deying@openkedge.io
}

\begin{document}

\maketitle

\nocopyright

\begin{abstract}
Automated program synthesis lowers the cost of producing implementations but introduces a harder governance problem: determining which generated artifacts are admissible. Natural-language specifications are ambiguous, and example-based tests sample only part of the behavioral space. Used alone, neither provides a sufficient control boundary.

We introduce \textbf{Protocol-Driven Development (PDD)}, where the primary software artifact is a machine-enforceable protocol rather than code. We define a protocol as the triplet $\mathcal{P} = (\mathcal{S}, \mathcal{B}, \mathcal{O})$, specifying structural, behavioral, and operational invariants. Their conjunction defines the admissible implementation space of a software component.

Under PDD, implementations are replaceable realizations discovered through constrained search. An implementation is admitted only if it satisfies the protocol and produces a verifiable \emph{Evidence Chain} of compliance. Admission is grounded in protocol satisfaction and recorded evidence rather than trust in the generator.

For deployed systems, we extend the Evidence Chain into a \emph{Dynamic Evidence Ledger}. Runtime verifiers append signed observations, invariant checks, and violations to the ledger, allowing monitorable obligations to be continuously attested. This connects live failures back to the generation loop without granting the generator runtime authority.

Combining formal methods, property testing, runtime verification, policy-as-code, and software provenance, PDD defines a governance model for automated software engineering. Its organizing principle is that code is transient, while the protocol carries durable authority.
\end{abstract}


\section{Introduction}

Generative program synthesis has reduced the marginal cost of producing candidate implementations. Boilerplate, interface implementations, service scaffolds, and many routine code changes can now be produced quickly by synthesis engines. This shift moves the central pressure from production to admission, and from admission to continued accountability in operation. The difficult question is no longer only how to produce code, but how to decide which generated implementations are structurally valid, behaviorally correct, operationally bounded, eligible for regeneration, and still compliant after deployment.

Traditional software engineering methodologies were developed when implementation effort dominated many workflows. In \emph{Spec-Driven Development} (SDD), natural-language documents describe intended structure and behavior, but they are often too ambiguous for automated admission. Different engineers---or different generative systems---may interpret the same requirement incompatibly, producing semantic drift, inconsistent assumptions, and hidden side effects.

\emph{Test-Driven Development} (TDD) improves on this model by encoding expected behavior in executable tests~\cite{beck2003tdd}. Tests reduce ambiguity and increase confidence in functional behavior, but they are fundamentally extensional: they verify selected cases rather than defining the full space of admissible implementations. They rarely express structural contracts, authority boundaries, or operational constraints such as latency limits, side-effect restrictions, dependency controls, or resource quotas.

Automated code generation makes these weaknesses harder to ignore. When implementations can be synthesized and regenerated at low cost, source code becomes less durable than the constraints that determine whether the code is admissible. The operative question changes from ``How do we implement this component?'' to ``What constraints must every acceptable implementation satisfy?'' In this setting, software engineering becomes less a matter of prescribing one implementation and more a matter of defining the invariant boundary within which implementations may be safely discovered.

We call this model \textbf{Protocol-Driven Development (PDD)}. In PDD, the primary artifact is a machine-enforceable protocol rather than implementation code. We formally define a protocol as
\[
\mathcal{P} = (\mathcal{S}, \mathcal{B}, \mathcal{O}),
\]
where $\mathcal{S}$ denotes \emph{structural invariants}, $\mathcal{B}$ denotes \emph{behavioral invariants}, and $\mathcal{O}$ denotes \emph{operational invariants}. Structural invariants define rigid type and interface contracts through typed handshakes. Behavioral invariants specify semantic properties that must hold for all admissible implementations. Operational invariants impose explicit capability boundaries covering side effects, latency, external dependencies, and resource consumption. Their conjunction defines the admissible implementation space within which automated generators may explore.

Under PDD, implementations are treated as replaceable realizations of a governing protocol. Automated generators may generate, discard, and regenerate candidate code, but a candidate is admissible only if it satisfies the protocol. We propose a \emph{Validator Loop} that performs structural validation, property-based verification, and operational conformance checks before admission. Admitted implementations then produce an \emph{Evidence Chain}: a cryptographically verifiable record of the governing protocol, implementation metadata, validation results, and observed behavioral characteristics. The Evidence Chain supports auditability, accountability, and replay of admission decisions when replay inputs are preserved.

Pre-deployment evidence is not the end of governance. Production traffic can expose memory leaks, race conditions, dependency drift, adversarial inputs, and operational degradation that were absent from simulation. PDD therefore extends Evidence Chains into \emph{Dynamic Evidence Ledgers}: append-only records that include signed runtime observations, invariant checks, violations, and remediation outcomes. Runtime verification layers enforce protocol boundaries outside the generated implementation, and violation blocks can be converted into structured repair contexts for the next generation attempt.

For concreteness, Appendix~\ref{app:minimal-pdd-bundle} sketches a minimal protocol bundle, typed handshake, behavioral invariant set, capability manifest, and evidence record.

We term the recurring cost of this ambiguity the \emph{Natural Language Tax}: the cost of interpreting, reconciling, and maintaining ambiguous textual descriptions. Rather than describing desired behavior in prose and relying on downstream interpretation, PDD encodes invariants directly as machine-enforceable assertions. Typed handshakes address type drift, property-based assertions address intent drift, and capability manifests address side-effect drift. The protocol becomes the authoritative admission artifact, and valid implementations are constrained to remain within its boundaries.

The model draws on formal methods, property-based testing, zero-trust security, and generative systems~\cite{hoare1969axiomatic,clarke1999model,claessen2000quickcheck,rose2020zerotrust}. It combines invariants, executable properties, explicit verification, and constrained search into a development model for software synthesis.

PDD does not claim to invent interface schemas, property-based testing, policy enforcement, or provenance and attestation mechanisms. Its novelty is the elevation of the protocol to the primary development artifact: a single governing object that jointly specifies structure, behavior, operational authority, and evidence-producing admission for generated software. Schema-first or interface-first development stabilizes component boundaries but usually leaves behavior, authority, and admission evidence outside the interface artifact. TDD and property-based testing constrain selected behavioral properties but do not by themselves define operational capability boundaries or provenance-linked admission. Policy-as-code governs authority but is typically external to software construction; supply-chain provenance records where artifacts came from rather than defining which artifacts are admissible. PDD unifies these concerns into a protocol-centered development model in which generated implementations are proposals admitted only through invariant satisfaction and verifiable evidence.

PDD also relates to runtime governance for autonomous systems while remaining independent of any particular deployment stack. Sovereign Agentic Loops~\cite{he2026sal} and OpenKedge~\cite{he2026openkedge} separate model-generated proposals from privileged execution through control boundaries, policy evaluation, and evidence records. PDD applies the same proposal-admission structure to software construction and then carries it into operation through runtime attestation.

The thesis of this work is:

\begin{quote}
\emph{Code is transient; protocol is sovereign.}
\end{quote}

By making protocols the durable engineering artifact, PDD treats software creation as constrained discovery: implementations may change while admissibility conditions remain explicit, testable, and auditable.

\vspace{0.5em}
\noindent\textbf{Contributions.}
The paper makes the following contributions:

\begin{enumerate}
    \item We propose Protocol-Driven Development (PDD), a model that treats protocols, rather than implementation code, as the durable artifact of automated software engineering.
    \item We formalize a protocol as a triplet of structural, behavioral, and operational invariants that together define an admissible implementation space for generated software.
    \item We define protocol compliance, protocol-level substitutability, and evidence-producing acceptance as the core relations governing admission under PDD.
    \item We outline the Validator Loop, an architectural pattern that separates protocol authoring, implementation generation, validation, and evidence preservation.
    \item We introduce Evidence Chains and Discovery Logs as artifacts for linking candidate implementations to validation outcomes, provenance, admission decisions, and later runtime observations.
    \item We extend Evidence Chains into Dynamic Evidence Ledgers, define continuous protocol attestation, and introduce a Runtime Verification Layer for enforcing invariants after deployment.
    \item We specify a closed-loop remediation path in which signed runtime violation blocks become repair context for subsequent generation and validation.
    \item We characterize the Natural Language Tax and argue that machine-enforceable protocols provide a disciplined way to move from narrative intent to explicit admissibility conditions.
    \item We synthesize connections to formal methods, property-based testing, runtime verification, policy-as-code, software provenance, automated software engineering, and runtime governance for autonomous systems.
\end{enumerate}

\section{Background and Motivation}

Protocol-Driven Development (PDD) sits at the intersection of specification, testing, formal verification, policy enforcement, and automated program synthesis. Each tradition addresses part of the correctness problem, but none by itself provides an admission model for generated implementations.

\subsection{Spec-Driven Development}

Spec-Driven Development (SDD) treats the specification as the prior artifact that guides implementation. Specifications range from prose requirements to machine-readable interface descriptions such as OpenAPI, JSON Schema, and Protocol Buffers~\cite{openapi2021,jsonschema2022,protobuf2008}. These artifacts establish a shared vocabulary and support automation such as documentation, client stubs, scaffolds, validation, and interface conformance testing.

SDD, however, leaves a semantic gap between what is written and what must be enforced. A requirement such as ``return the current user profile'' rarely specifies determinism, retry behavior, external calls, or latency budgets. Schema languages reduce structural ambiguity, but they usually do not capture full behavioral and operational semantics. PDD begins from this gap: the artifact that governs generated software must be more restrictive than prose and more expressive than schema alone.

\subsection{Test-Driven Development}

Test-Driven Development (TDD) moves part of engineering intent into executable form. Tests expose assumptions, define concrete acceptance criteria, and reject regressions automatically~\cite{beck2003tdd}. Fuzzing and property-based testing extend this idea by exploring larger behavior spaces through generated inputs~\cite{claessen2000quickcheck}.

Yet tests remain partial witnesses. They may explore examples or input families, but they do not necessarily define the full protocol by which a module is permitted to exist inside a larger system. Tests also tend to focus on functional behavior while leaving authority boundaries implicit. In PDD, tests are one mechanism for validating behavioral invariants, not the complete engineering artifact.

\subsection{Formal Methods}

Formal methods provide an important precedent for PDD. Hoare logic, model checking, TLA+, and Alloy show that invariants, preconditions, postconditions, temporal properties, and relational constraints can expose design errors before implementation~\cite{hoare1969axiomatic,clarke1999model,lamport2002specifying,jackson2002alloy}.

PDD does not replace formal methods; it changes where formalization sits. A protocol can include temporal specifications where appropriate, and it can combine typed interfaces, executable properties, and capability manifests within a single admission artifact. PDD does not derive every implementation from first principles; it requires every admitted implementation to produce evidence that it satisfies the constraints needed for composition.

\subsection{Zero Trust, Policy Systems, and Provenance}

Zero-trust architecture rejects implicit trust based on location, identity, or reputation, and instead requires explicit verification~\cite{rose2020zerotrust}. Policy-as-code systems such as Open Policy Agent decouple authorization logic from application code~\cite{opa2016}. Supply-chain frameworks such as SLSA and in-toto record provenance and attestations so that artifacts can be evaluated through evidence rather than assumption~\cite{slsa2021,torresarias2019intoto}.

PDD brings that verification stance into software construction. An implementation generator proposes code, but the code must satisfy a protocol under validation. Operational invariants restrict what a module may do regardless of whether it appears functionally correct, and Evidence Chains make admission auditable rather than dependent on informal confidence.

\subsection{Automated Software Engineering}

The shift toward generated software was anticipated by the ``Software 2.0'' framing, which treats parts of software construction as search over program space rather than direct manual authorship~\cite{karpathy2017software}. Code-generating models synthesize programs from natural-language prompts~\cite{chen2021codex}, generative assistants can improve productivity on selected tasks~\cite{peng2023copilot}, and repository-level agents and benchmarks make long-horizon automated software engineering measurable~\cite{jimenez2024swebench,yang2024sweagent,wu2024devin}.

Low-cost generation makes this problem concrete. The same prompt may produce different structures, dependencies, resource behaviors, and failure modes across models, temperatures, versions, and tool environments. PDD addresses this risk by making the prompt non-authoritative: natural language may initiate development, but the protocol determines admission.

\subsection{Why a New Model Is Needed}

The prior traditions reveal a consistent pattern. Specifications communicate intent but remain ambiguous; tests provide executable evidence but remain partial; formal methods provide rigor but are often selective; policy systems govern authority but usually at runtime; automated coding systems accelerate implementation but amplify variation.

This leaves a gap: a governing artifact that unifies structure, semantics, and operational authority into an admissibility boundary for generated software. PDD addresses this gap by treating implementation as a search outcome and protocol as the durable object that defines which software artifacts are admissible.

The motivating claim of PDD is not that prior models were wrong. It is that they are incomplete for automated program synthesis. When implementations are inexpensive to produce and easier to regenerate, the enduring challenge is not only how to describe or test one program, but how to define the invariants that every acceptable program must obey.

\section{Protocol-Driven Development Model}

We formalize PDD in terms of three commitments: the protocol is the primary engineering artifact, compliance gives the admission criterion, and software construction is constrained search over a protocol-defined implementation space.

\subsection{The Protocol as Governing Artifact}

PDD starts from the premise that implementation code is not the durable representation of a software component. The durable representation is a machine-enforceable protocol specifying the invariant boundaries within which implementations are discovered, validated, and substituted. The protocol governs admission; implementations remain replaceable realizations.

\begin{definition}[Protocol]
A \emph{protocol} is a triplet
\[
\mathcal{P} = (\mathcal{S}, \mathcal{B}, \mathcal{O}),
\]
where:
\begin{itemize}
    \item $\mathcal{S}$ is the set of \emph{structural invariants};
    \item $\mathcal{B}$ is the set of \emph{behavioral invariants};
    \item $\mathcal{O}$ is the set of \emph{operational invariants}.
\end{itemize}
\end{definition}

The protocol defines the admissibility boundary for a software component. Any implementation that satisfies all three invariant classes is protocol-compliant, regardless of the algorithms, programming language, or internal organization used to realize it. The same triplet also defines the observation boundary for deployed instances: runtime evidence may evolve, but the protocol remains the governing object against which new observations are checked.

\subsection{Structural Invariants}

Structural invariants specify the syntactic and interface-level constraints governing communication between components.

\begin{definition}[Structural Invariants]
Structural invariants $\mathcal{S}$ define rigid type and schema constraints over input and output representations, interface signatures, field types and cardinalities, and versioning rules.
\end{definition}

Examples include Protocol Buffers, TypeScript interfaces, JSON Schema, and OpenAPI specifications. Structural invariants ensure that all compliant implementations expose the same typed handshake.

\subsection{Behavioral Invariants}

Behavioral invariants specify the semantic properties that must hold across admissible executions.

\begin{definition}[Behavioral Invariants]
Behavioral invariants $\mathcal{B}$ are predicates over observable behavior that must hold for all valid inputs and admissible states under the protocol's observation model.
\end{definition}

Examples include idempotence, determinism, monotonicity, conservation laws, and error propagation guarantees. These invariants can be encoded as property-based tests, logical assertions, metamorphic relations, or formal specifications.

\subsection{Operational Invariants}

Operational invariants bound the capabilities, side effects, and resource consumption of an implementation.

\begin{definition}[Operational Invariants]
Operational invariants $\mathcal{O}$ specify admissible operational behavior, including constraints on external calls, dependency usage, disk and network access, execution latency, memory consumption, and concurrency.
\end{definition}

Examples include egress allowlists, latency and memory bounds, and dependency restrictions. They function as capability manifests or policy boundaries, restricting operational authority independently of functional correctness.

\subsection{Modularity, Testability, and Operational Boundedness}

The three invariant classes correspond to three governance properties. Structural invariants support \emph{modularity} through typed handshakes; behavioral invariants support \emph{testability} through executable or checkable predicates; operational invariants support \emph{boundedness} by restricting authority, side effects, and resource usage. A candidate must satisfy all three to be admissible.

\subsection{Implementation Space}

Let $\mathcal{I}$ denote the universe of candidate implementations expressible within a given programming environment.

\begin{definition}[Admissible Implementation Set]
Given protocol $\mathcal{P}$, the admissible implementation set is
\[
\mathcal{I}_{\mathcal{P}}
=
\{\, I \in \mathcal{I} \mid I \models \mathcal{P} \,\}.
\]
\end{definition}

Equivalently, if $\mathcal{I}_{\mathcal{S}}$, $\mathcal{I}_{\mathcal{B}}$, and $\mathcal{I}_{\mathcal{O}}$ denote the sets of implementations satisfying the respective invariant classes, then
\[
\mathcal{I}_{\mathcal{P}}
=
\mathcal{I}_{\mathcal{S}}
\cap
\mathcal{I}_{\mathcal{B}}
\cap
\mathcal{I}_{\mathcal{O}}.
\]

Thus, a protocol defines an admissible region of implementation space, not a single implementation.

\subsection{Compliance Relation}

Protocol compliance is defined as conjunction across the three invariant classes.

\begin{definition}[Protocol Compliance]
An implementation $I$ is \emph{protocol-compliant} if and only if
\[
I \models \mathcal{P}
\iff
\bigl(
I \models \mathcal{S}
\bigr)
\land
\bigl(
I \models \mathcal{B}
\bigr)
\land
\bigl(
I \models \mathcal{O}
\bigr).
\]
\end{definition}

The definition separates admissibility from implementation strategy: internally different implementations are equally compliant when they satisfy the same protocol-visible constraints.

\subsection{Development as Constrained Search}

Under Protocol-Driven Development, software construction is modeled as search over $\mathcal{I}_{\mathcal{P}}$.

An implementation generator explores candidate realizations
\[
I_1, I_2, \dots, I_n \in \mathcal{I},
\]
until it discovers an implementation $I_k$ such that
\[
I_k \models \mathcal{P}.
\]

The implementation is accepted because it satisfies the governing constraints, not because it is manually authored, generator-preferred, or trusted by origin.

\subsection{Protocol-Level Substitutability}

A defining property of PDD is that implementations are replaceable with respect to the commitments made by the protocol. Compliant implementations need not be observationally identical. A protocol deliberately leaves algorithmic choices, internal data structures, or performance strategies unspecified when they are outside its commitment surface. Substitutability holds for clients whose assumptions are limited to the protocol's structural, behavioral, and operational guarantees.

\begin{definition}[Protocol-Level Substitutability]
Two implementations $I_a$ and $I_b$ are \emph{substitutable under protocol} $\mathcal{P}$ for a protocol-respecting client if both implementations satisfy the protocol,
\[
I_a \models \mathcal{P}
\quad \text{and} \quad
I_b \models \mathcal{P},
\]
and the client depends only on guarantees entailed by $\mathcal{P}$.
\end{definition}

Such implementations can differ internally while remaining interchangeable with respect to the guarantees the protocol actually makes. This supports regeneration, optimization, refactoring, and language migration when downstream components depend on the protocol rather than incidental behavior.

\subsection{Evidence-Producing Acceptance}

Compliance must be established through machine-verifiable evidence. We use \emph{Evidence Chain} to denote the ordered record linking a protocol bundle, a candidate implementation, validator execution, Discovery Logs, signed attestations, and the resulting admission decision. For deployed implementations, the chain can be extended into a Dynamic Evidence Ledger that appends runtime attestations and violation blocks while preserving the original build-time admission evidence.

\begin{definition}[Validation Function]
Let $\mathcal{E}$ be the space of valid evidence objects. Let
\[
\mathrm{Validate}(I,\mathcal{P}) \in \mathcal{E} \cup \{\bot\}
\]
be a validation function that returns an evidence object $E \in \mathcal{E}$ when the validator establishes $I \models \mathcal{P}$, and returns $\bot$ when admission fails.
\end{definition}

The evidence object $E$ is one signed element of the Evidence Chain, typically containing protocol versions, implementation hashes, validation outputs, and validator attestations. If validation fails, no acceptance evidence is produced. After deployment, successful and failed runtime attestations are not admission events for the current artifact; they are evidence about whether the admitted artifact continues to satisfy the monitorable projection of $\mathcal{P}$ under live observations.

\subsection{Summary}

In PDD, software construction centers on the protocol, the object that combines structural, behavioral, and operational invariants into an admissible implementation space. The protocol, rather than any specific implementation, becomes the authoritative admission contract for a component. Implementations are accepted only when accompanied by machine-verifiable evidence of compliance.

\section{Protocol Authoring and Disambiguation}

Protocol-Driven Development begins by translating architectural intent into machine-enforceable constraints. Informal requirements, design discussions, or prompts can initiate the process, but the authoritative artifact is the protocol: a specification whose admissibility conditions are machine-validatable.

Protocol authoring employs three mechanisms: \emph{typed handshakes} for structure, \emph{property-based assertions} for behavior, and \emph{capability manifests} for operational authority. Together they map prose to restrictive assertions over structure, semantics, and side effects.

\subsection{The Natural Language Tax}

The \emph{Natural Language Tax} is the recurring cost imposed when durable engineering intent is represented primarily in prose. It appears as \emph{interpretation cost}, when developers or agents infer types, edge cases, capability boundaries, and failures from underspecified statements; \emph{reconciliation cost}, when incompatible interpretations are repaired after implementation begins; and \emph{maintenance cost}, when prose drifts as implementations evolve.

PDD reduces this tax by making prose a staging artifact rather than the final control surface. Natural language may initiate protocol authoring, but the durable output is a set of machine-enforceable constraints: typed handshakes for representation, property-based assertions for behavior, and capability manifests for side effects and authority.

\subsection{Typed Handshakes}

A PDD protocol begins with a \emph{typed handshake}: a machine-readable contract specifying inputs, outputs, errors, and compatibility rules. Typed handshakes can be written in JSON Schema, OpenAPI, Protocol Buffers, TypeScript interfaces, or equivalent schema languages~\cite{jsonschema2022,openapi2021,protobuf2008}.

A typed handshake reduces \emph{type drift} between architectural intent, generated implementation, and downstream consumers. Instead of saying ``returns a user object,'' PDD defines a versioned schema with required and optional fields, enumerations, nullability, error variants, and deprecation rules. Ambiguity becomes a compile-time or validation failure.

As long as the handshake is preserved, protocol-compliant implementations remain interchangeable at the structural boundary.

\subsection{Property-Based Assertions}

Protocol authoring next elevates example-based expectations into general behavioral laws. Instead of stating informally that ``the handler should be idempotent,'' the protocol encodes the invariant directly as
\[
\forall x \in X,\quad f(f(x)) = f(x).
\]
Similarly, rather than saying that ``invalid input should fail safely,'' the protocol defines
\[
\mathrm{invalid}(x) \Rightarrow \mathrm{isError}(f(x)).
\]

These properties become validator targets for property-based testing frameworks, symbolic analysis tools, theorem provers, or runtime validation mechanisms. Regardless of the enforcement mechanism, the invariant itself becomes the authoritative statement of intended behavior.

Property-based assertions reduce \emph{intent drift}: they prevent an implementation from satisfying visible examples while violating the general property those examples were meant to express. In PDD, the acceptance criterion is not a curated test list, but the invariant itself.

\subsection{Capability Manifests}

The third step specifies what the implementation is allowed to do. A \emph{capability manifest} defines operational authority and resource boundaries, including:

\begin{itemize}
    \item file-system access;
    \item outbound network destinations;
    \item database operations and transaction limits;
    \item maximum external calls per request;
    \item memory and CPU budgets;
    \item latency targets;
    \item environment variables and secret access;
    \item concurrency bounds;
    \item permitted background work.
\end{itemize}

These constraints become enforcement targets for sandboxing, runtime instrumentation, policy engines, operating system controls, or deployment configuration.

Capability manifests reduce \emph{side-effect drift}. An implementation may pass functional tests while introducing hidden caches, third-party API calls, temporary file writes, or expanded transaction scopes. Under PDD, such behavior is part of the protocol: if a module has not been granted disk I/O, any write attempt is a violation regardless of functional output.

\subsection{Evidence Chains}

The final mechanism bridges the gap between protocol definition and verifiable acceptance. An \emph{Evidence Chain} is a signed record that binds a specific implementation artifact to the protocol it satisfies, the validator that assessed it, and the Discovery Log of its operational characteristics.

Evidence Chains reduce \emph{audit drift}: the divergence between running software and the rationale for why it was admitted. Instead of relying on statements such as ``it passed all tests locally,'' the Evidence Chain preserves a machine-checkable record of compliance.

\subsection{Constraining Ambiguity}

Protocol-Driven Development replaces descriptive statements with restrictive assertions. Human authors can still begin with prose, but prose is not authoritative. Each common class of ambiguity is mapped to a corresponding machine-enforceable mechanism, as summarized in Table~\ref{tab:disambiguation}.

\begin{table*}[t]
\centering
\caption{Protocol-Driven Development mapping of ambiguity classes.}
\label{tab:disambiguation}
\renewcommand{\arraystretch}{1.08}
\begin{tabular}{@{}p{0.22\textwidth}p{0.52\textwidth}p{0.18\textwidth}@{}}
\toprule
\textbf{Mechanism} & \textbf{Illustrative Narrative Requirement} & \textbf{Drift Removed} \\
\midrule
Typed handshake & ``Returns a user object'' & Type drift \\
Property-based assertion & ``Handles invalid input safely'' & Intent drift \\
Capability manifest & ``Do not call the database too much'' & Side-effect drift \\
Evidence Chain & ``It passed on my machine'' & Audit drift \\
\bottomrule
\end{tabular}
\end{table*}

The mapping moves ambiguity from social interpretation into enforceable structure. The architect need not prescribe every implementation detail; the architect defines the admissibility conditions that all valid implementations must satisfy.

\subsection{From Narrative to Law}

PDD moves from descriptive intent to formal constraint. Once structural, behavioral, and operational invariants are encoded as machine-enforceable constraints, implementations can be generated, discarded, and regenerated without changing the authoritative artifact. Protocol authoring therefore becomes the main design activity: the architect defines the invariant boundaries within which acceptable implementations are discovered.

\section{Validator Loop and Evidence Chains}

The \emph{Validator Loop} is the admission cycle for PDD. It separates protocol definition, candidate generation, validation, and admission. The generator searches; the protocol and validator define success. In deployed systems, the same evidence discipline continues through runtime attestation, but runtime observations do not bypass admission: repaired or regenerated artifacts still return to the Validator Loop before replacement.

Protocol Authors define admissibility conditions, Implementation Generators propose candidates, Validation Engines evaluate candidates, and Evidence Stores preserve the basis for acceptance. A software artifact becomes admissible only after compliance is evaluated and recorded.

\subsection{Contract Negotiation}

The Validator Loop begins when a \emph{Protocol Author} proposes a protocol $\mathcal{P}$. Because modules rarely exist in isolation, dependencies must be reconciled into a coherent admissibility boundary before implementation begins.

\emph{Contract negotiation} verifies that the typed handshakes, behavioral properties, and operational capabilities of dependent protocols are mutually compatible. It includes dependency resolution, compatibility checking, capability reconciliation, and conflict detection across transitive protocol boundaries.

The output of negotiation is a \emph{versioned protocol bundle} containing:

\begin{itemize}
    \item structural schemas and interface definitions;
    \item behavioral properties and regression constraints;
    \item operational capability manifests;
    \item protocol dependencies and compatibility metadata;
    \item approved validator implementations and versions.
\end{itemize}

Once sealed, the bundle becomes the target for generation. Generators are not permitted to silently weaken or reinterpret its constraints; protocol changes require an explicit version event and renewed negotiation.

The illustrative bundle in Appendix~\ref{app:minimal-pdd-bundle} shows one compact representation of these artifacts without prescribing a particular format.

\subsection{Automated Generation}

An \emph{Implementation Generator} receives the protocol bundle and searches for a candidate implementation using prompting, retrieval, synthesis, repair, evolutionary search, templates, or other strategies. PDD is agnostic to model, language, and toolchain.

A candidate implementation $I$ is treated as untrusted until it has passed validation against the protocol bundle.

Formally, the Implementation Generator explores a sequence
\[
I_1, I_2, \ldots, I_n \in \mathcal{I},
\]
until it discovers some implementation $I_k$ such that
\[
I_k \models \mathcal{P}.
\]

The loop permits multiple generators to compete against the same protocol and regenerated implementations to replace earlier ones without changing dependent contracts. The implementation is a provisional witness to satisfiability, not the durable artifact.

\subsection{Verification}

A \emph{Validation Engine} performs verification: it inspects artifacts, executes validation logic, monitors resources, and rejects non-compliant candidates. It acts as the admission controller for the implementation space.

Verification proceeds in three layers:

\begin{enumerate}
    \item \textbf{Structural validation} checks that interfaces compile, serialize, and conform to $\mathcal{S}$.
    \item \textbf{Behavioral validation} checks properties, examples, metamorphic relations, and regression suites against $\mathcal{B}$.
    \item \textbf{Operational validation} executes the candidate under policy, sandboxing, or instrumentation to verify compliance with $\mathcal{O}$.
\end{enumerate}

The layers are jointly necessary: structural checks alone miss semantics, behavioral checks leave hidden capability violations unobserved, and operational checks cannot establish functional meaning.

Because validators define the acceptance boundary, they must be trusted at least as much as the build and release system. In high-assurance settings, validators may themselves be versioned, sandboxed, reproducible, and attested.

\subsection{Discovery Logs}

When validation succeeds, the admitted implementation emits a \emph{Discovery Log}: an as-built record of what was produced and observed. The protocol states what \emph{must} be true; the Discovery Log records what \emph{was found} to be true of this implementation.

A Discovery Log includes entries such as:

\begin{itemize}
    \item implementation language and compiler versions;
    \item dependency graph and package hashes;
    \item generated files and artifact digests;
    \item validator identities and versions;
    \item property coverage and validation outcomes;
    \item observed resource usage;
    \item derived behaviors not explicitly enumerated in the original protocol.
\end{itemize}

Discovery Logs make acceptance auditable and provide feedback for protocol evolution. Useful recurring behaviors can be promoted into future protocol versions; undesirable behaviors can motivate stronger constraints.

\subsection{Evidence Chains}

An \emph{Evidence Chain} is the ordered record linking protocol constraints, generated implementations, validation outcomes, and deployment artifacts. Within PDD, it binds software admission to accountable evidence.

Let the validator produce an evidence object
\[
E = H(\mathcal{P}, I, V, R, t),
\]
where:
\begin{itemize}
    \item $\mathcal{P}$ is the negotiated protocol bundle;
    \item $I$ is the admitted implementation artifact;
    \item $V$ identifies the validator implementations and versions;
    \item $R$ contains validation results and measured observations;
    \item $t$ records time, environment, and provenance metadata;
    \item $H$ denotes a cryptographic digest or signed attestation over these elements.
\end{itemize}

The evidence object supports audit, end-to-end accountability, and replay whenever the relevant validator inputs are preserved. A downstream system asks whether the artifact is linked to an approved protocol and whether the Evidence Chain records satisfaction under approved validators.

The same evidence structure links development and runtime governance when deployment systems consume development-time evidence before runtime admission. Once an implementation is deployed, runtime monitors can append additional evidence blocks to the same logical record. These blocks may attest continued compliance, record bounded degradation, or capture invariant violations for later repair. Section~\ref{sec:runtime-evidence} formalizes this extension as a Dynamic Evidence Ledger.

\subsection{Acceptance as an Evidence-Producing Event}

In PDD, validation is the mechanism by which software becomes admissible. An implementation is accepted if and only if:

\begin{enumerate}
    \item a protocol bundle has been negotiated and sealed;
    \item a candidate implementation has been generated;
    \item the candidate satisfies all structural, behavioral, and operational invariants; and
    \item the validator emits signed evidence linking the artifact to the governing protocol.
\end{enumerate}

The admission principle is:

\begin{quote}
\emph{Code is accepted not because it appears correct, but because compliance has been validated and recorded.}
\end{quote}

Under this model, each admitted implementation carries a verifiable chain showing why it was permitted to enter the system.

\subsection{Summary}

The Validator Loop separates proposal, verification, and admission. Protocol bundles define admissibility, generators explore candidates, validators establish compliance, Discovery Logs record observations, and Evidence Chains preserve the basis for admission.

\section{Theoretical Foundations}

The formal account below focuses on the core PDD relations. It does not prove that every useful software property is decidable or that validation is complete. Rather, it states what follows when protocols define the relevant observation boundary and validators are sound with respect to that boundary.

\subsection{Implementation-Space Interpretation}

Let $\mathcal{I}$ denote the set of candidate implementations expressible by a generator within a target environment. The set includes programs produced through prompting, retrieval, synthesis, repair, mutation, or any other search strategy. A natural-language prompt defines at most an imprecise region of $\mathcal{I}$: different models, sampling temperatures, prompts, or tool contexts may produce implementations that appear plausible while diverging in structure, behavior, dependencies, and operational footprint.

PDD replaces this open-ended search with an explicit admissibility boundary. For a protocol
\[
\mathcal{P}=(\mathcal{S},\mathcal{B},\mathcal{O}),
\]
let $\mathcal{I}_{\mathcal{S}}$, $\mathcal{I}_{\mathcal{B}}$, and $\mathcal{I}_{\mathcal{O}}$ denote the implementations satisfying the structural, behavioral, and operational invariant classes, respectively. The admissible implementation set induced by $\mathcal{P}$ is
\[
\mathcal{I}_{\mathcal{P}}
=
\mathcal{I}_{\mathcal{S}}
\cap
\mathcal{I}_{\mathcal{B}}
\cap
\mathcal{I}_{\mathcal{O}}.
\]

Thus, PDD models software construction as constrained search from plausible implementations to admissible implementations. The protocol does not define a single program; it defines the region of implementation space within which candidate programs are admitted.

\subsection{Observation-Model Semantics}

The satisfaction relation $I \models \mathcal{P}$ is observational rather than intensional. Compliance is not defined by the private internal organization of $I$, but by the observations that $\mathcal{P}$ makes relevant.

\begin{definition}[Protocol Observation Model]
For protocol $\mathcal{P}$, let $\Omega_{\mathcal{P}}$ be the set of protocol-relevant observations and let $\Phi_{\mathcal{P}}$ be the predicate that characterizes admissible observations. For implementation $I$, let $\mathrm{Obs}_{\mathcal{P}}(I) \subseteq \Omega_{\mathcal{P}}$ denote the observations of $I$ visible under the protocol's observation model.
\end{definition}

The observation set includes the protocol-relevant structural, behavioral, and operational observations: schema conformance and serialization behavior; outputs, errors, invariants, and temporal relations; and resource use, side effects, external calls, and capability boundaries. The protocol predicate decomposes as
\[
\Phi_{\mathcal{P}}(\omega)
\iff
\Phi_{\mathcal{S}}(\omega)
\land
\Phi_{\mathcal{B}}(\omega)
\land
\Phi_{\mathcal{O}}(\omega),
\]
where each component corresponds to one invariant class.

\begin{definition}[Protocol Satisfaction]
An implementation $I$ satisfies protocol $\mathcal{P}$, written $I \models \mathcal{P}$, if and only if
\[
\forall \omega \in \mathrm{Obs}_{\mathcal{P}}(I),\quad \Phi_{\mathcal{P}}(\omega).
\]
\end{definition}

This definition permits heterogeneous implementations. A Python service, a Rust binary, and a generated WebAssembly module can all satisfy the same protocol when their protocol-visible observations satisfy the same structural, behavioral, and operational predicates. Outside the protocol's commitment surface, they need not be identical.

\subsection{Validators and Evidence-Carrying Admission}

A validator is the mechanism that decides whether a candidate enters $\mathcal{I}_{\mathcal{P}}$. We write
\[
\mathrm{Validate}(I,\mathcal{P}) \in \mathcal{E} \cup \{\bot\},
\]
where $\mathcal{E}$ is the space of valid evidence objects and $\bot$ denotes rejection.

\begin{definition}[Validator Soundness]
A validator for protocol $\mathcal{P}$ is sound if, whenever $\mathrm{Validate}(I,\mathcal{P})=E$ with $E \neq \bot$, it follows that $I \models \mathcal{P}$.
\end{definition}

\begin{theorem}[Evidence-Carrying Sound Admission]
Assume a sound validator for $\mathcal{P}$ and an evidence object $E \in \mathcal{E}$ that cryptographically binds the protocol identity, implementation identity, validator identity, and validation result. If $\mathrm{Validate}(I,\mathcal{P})=E$ and $E \neq \bot$, then $I \in \mathcal{I}_{\mathcal{P}}$. Moreover, the evidence cannot be interpreted as an admission of a different implementation or a different protocol without breaking the binding assumption.
\end{theorem}

\begin{proof}[Proof]
Since the validator is sound and returns $E \neq \bot$, we have $I \models \mathcal{P}$. By the definition of the admissible implementation set, $I \in \mathcal{I}_{\mathcal{P}}$. The second claim follows from the assumed binding of $E$ to the protocol identity, implementation identity, validator identity, and validation result.
\end{proof}

The theorem is intentionally conditional. If the validator is unsound, if the protocol omits a relevant property, or if the evidence binding is compromised, PDD does not provide the stated admission guarantee. The formal claim is that evidence-bearing admission is meaningful only relative to a specified protocol and a sound validation boundary.

\subsection{Protocol-Respecting Clients}

Substitutability in PDD is not full observational equivalence; it is equivalence relative to the guarantees exposed by the protocol. Let $G(\mathcal{P})$ denote the set of guarantees entailed by $\mathcal{P}$ under its observation model.

\begin{definition}[Protocol-Respecting Client]
A client $C$ is protocol-respecting with respect to $\mathcal{P}$ if every assumption $C$ makes about a component is contained in $G(\mathcal{P})$. Such a client is permitted to rely on the protocol's structural, behavioral, and operational guarantees, but not on implementation internals or behavior outside the protocol's commitment surface.
\end{definition}

\begin{theorem}[Protocol-Level Substitutability]
Let $I_a,I_b \in \mathcal{I}_{\mathcal{P}}$. For any client $C$ that is protocol-respecting with respect to $\mathcal{P}$, replacing $I_a$ with $I_b$ preserves every client obligation that depends only on $G(\mathcal{P})$.
\end{theorem}

\begin{proof}[Proof]
Since $I_a,I_b \in \mathcal{I}_{\mathcal{P}}$, both satisfy $\mathcal{P}$ and therefore satisfy the guarantees in $G(\mathcal{P})$. Because $C$ is protocol-respecting, its component-facing assumptions are limited to $G(\mathcal{P})$. Replacing $I_a$ with $I_b$ therefore preserves all assumptions on which $C$ is permitted to depend.
\end{proof}

This result does not claim that $I_a$ and $I_b$ are identical in latency, internal state layout, dependency choices, or behavior that the protocol leaves unspecified. Such differences are admissible precisely when they lie outside the protocol's commitment surface.

\begin{theorem}[Safe Regeneration Under Protocol-Respecting Dependency]
Let $I_{\mathrm{old}}$ be an implementation admitted for protocol $\mathcal{P}$ by a sound validator, and let $I_{\mathrm{new}}$ be a regenerated implementation such that $I_{\mathrm{new}} \in \mathcal{I}_{\mathcal{P}}$. If all downstream dependencies are protocol-respecting with respect to $\mathcal{P}$, then replacing $I_{\mathrm{old}}$ with $I_{\mathrm{new}}$ preserves all downstream obligations expressible in $G(\mathcal{P})$.
\end{theorem}

\begin{proof}[Proof]
Because $I_{\mathrm{old}}$ is admitted under $\mathcal{P}$ by a sound validator, evidence-carrying sound admission gives $I_{\mathrm{old}} \in \mathcal{I}_{\mathcal{P}}$. By assumption, $I_{\mathrm{new}} \in \mathcal{I}_{\mathcal{P}}$. Applying protocol-level substitutability to each protocol-respecting downstream dependency shows that replacement preserves every dependency obligation whose assumptions are contained in $G(\mathcal{P})$.
\end{proof}

Safe regeneration is therefore not a claim about arbitrary clients. It holds only for dependencies whose assumptions are constrained to the protocol. If a downstream component depends on accidental behavior of the old implementation, regeneration may expose that invalid coupling.

\subsection{Protocol Refinement}

Protocol refinement evolves protocols by adding constraints. We write $\mathcal{P}' \succeq \mathcal{P}$ when $\mathcal{P}'$ preserves every constraint of $\mathcal{P}$ and adds further structural, behavioral, or operational constraints. Equivalently, every observation admitted by $\mathcal{P}'$ is also admitted by $\mathcal{P}$.

\begin{proposition}[Refinement Narrows Admissibility]
If $\mathcal{P}' \succeq \mathcal{P}$, then
\[
\mathcal{I}_{\mathcal{P}'} \subseteq \mathcal{I}_{\mathcal{P}}.
\]
\end{proposition}

\begin{proof}[Proof]
For any implementation $I$, if $I \models \mathcal{P}'$, then all observations of $I$ satisfy the stronger predicate $\Phi_{\mathcal{P}'}$. Since $\mathcal{P}'$ preserves every constraint of $\mathcal{P}$, those observations also satisfy $\Phi_{\mathcal{P}}$, so $I \models \mathcal{P}$. Hence every implementation in $\mathcal{I}_{\mathcal{P}'}$ is also in $\mathcal{I}_{\mathcal{P}}$.
\end{proof}

Refinement captures the governance effect of protocol evolution. Adding constraints can exclude previously admissible implementations, but it cannot enlarge the admissible set unless the protocol is weakened rather than refined.

\subsection{Summary}

The formal role of PDD is to separate implementation search from artifact admission. Protocols define an observation boundary, validators provide conditional evidence that a candidate lies within that boundary, protocol-respecting clients can substitute or regenerate implementations without depending on incidental behavior, and refinement monotonically narrows the admissible implementation set. These results make precise the central claim of the paper: code is transient, while the protocol is the durable representation of software governance.

\section{Reference Architecture}

The reference architecture identifies minimal components for separating design intent, implementation search, validation, runtime enforcement, and evidence preservation without prescribing an implementation stack.

\subsection{Architectural Overview}

A minimal PDD system consists of seven principal components:

\begin{enumerate}
    \item \textbf{Protocol Author}, which authors and evolves protocols;
    \item \textbf{Protocol Registry}, which stores versioned protocol bundles;
    \item \textbf{Implementation Generator}, which searches for candidate realizations;
    \item \textbf{Validation Engine}, which verifies protocol compliance;
    \item \textbf{Runtime Verification Layer}, which observes and enforces protocol invariants during live execution;
    \item \textbf{Evidence Store}, which preserves Discovery Logs, Evidence Chain artifacts, and Dynamic Evidence Ledger blocks;
    \item \textbf{Remediation Orchestrator}, which converts runtime violations into repair contexts for renewed generation and validation.
\end{enumerate}

The components form a closed governance loop: protocols define admissibility, generators search, validators decide admission, runtime verifiers monitor deployed behavior, evidence stores record the basis of each judgment, and remediation returns failures to the generation pipeline.

The artifact sketch in Appendix~\ref{app:minimal-pdd-bundle} connects these roles to concrete protocol files, validator declarations, and evidence records without treating the sketch as a deployed prototype.

\subsection{Protocol Author}

The \emph{Protocol Author} translates intent into structural, behavioral, and operational invariants; negotiates dependencies; resolves compatibility conflicts; and versions protocol bundles. The role can be human, automated, or hybrid. Instead of prescribing one implementation, it defines the invariant boundary within which implementations are discovered and substituted.

\subsection{Protocol Registry}

The \emph{Protocol Registry} stores versioned protocol bundles: typed handshakes, behavioral invariants, capability manifests, dependency declarations, validator requirements, and provenance metadata. It is the durable system of record for protocol evolution.

\subsection{Implementation Generator}

The \emph{Implementation Generator} produces candidate realizations of a protocol bundle using generative models, synthesis, templates, or other search strategies. It operates outside the admission boundary: its outputs are proposals, and it has no authority to declare them valid.

\subsection{Validation Engine}

The \emph{Validation Engine} evaluates candidate artifacts against the governing protocol for structural, behavioral, and operational compliance. As required by the protocol, it invokes compilers, property-based testing, symbolic analyzers, sandbox monitors, policy engines, or provenance tooling. Generation proposes; validation decides.

\subsection{Runtime Verification Layer}

The \emph{Runtime Verification Layer} observes deployed executions from outside the generated implementation. It checks live payloads, state transitions, dependency calls, and resource measurements against the monitorable subset of the same structural, behavioral, and operational invariant classes used for admission. When an execution violates the protocol, the layer can block, quarantine, rate-limit, roll back, or emit a violation block for the Dynamic Evidence Ledger.

\subsection{Evidence Store}

The \emph{Evidence Store} preserves artifacts required to audit acceptance decisions and replay them when validator inputs are available: Discovery Logs, signed evidence objects, validator traces, implementation hashes, provenance records, runtime attestation blocks, and violation reports.

\subsection{Remediation Orchestrator}

The \emph{Remediation Orchestrator} turns runtime violation evidence into a structured repair context. It identifies the violated protocol clause, gathers relevant telemetry and environment metadata, and returns the failure to the Implementation Generator. The resulting patch or regenerated artifact is treated as a new candidate and must pass the Validation Engine before deployment.

\subsection{End-to-End Flow}

The end-to-end flow proceeds as follows:

\begin{enumerate}
    \item The Protocol Author defines and seals a protocol bundle.
    \item The Protocol Registry publishes the bundle as the authoritative component contract.
    \item The Implementation Generator retrieves the bundle and generates candidate realizations.
    \item The Validation Engine evaluates each candidate against the protocol.
    \item Upon successful validation, the system emits a Discovery Log and a signed evidence object.
    \item The Evidence Store records the resulting artifacts as part of the Evidence Chain.
    \item Approved implementations are released, deployed, or made available for substitution.
    \item The Runtime Verification Layer monitors deployed behavior and appends attestation blocks to the Dynamic Evidence Ledger.
    \item Runtime violations are converted into repair contexts, and any proposed fix re-enters validation before replacement.
\end{enumerate}

This treats software construction and operation as a continuous governance workflow.

\subsection{Trust Boundaries and Interoperability}

The reference architecture defines explicit admission and enforcement boundaries: generators and code artifacts are untrusted by default, protocols, validators, runtime verifiers, and evidence stores form the control plane, and deployment systems are conditionally trusted to verify evidence before runtime admission.

The architecture augments existing toolchains. Typed handshakes map to OpenAPI and Protocol Buffers, behavioral invariants integrate with property-based testing and formal verification, operational invariants use policy engines, runtime verifiers can be implemented with sidecars, gateways, service meshes, sandbox monitors, or kernel instrumentation, and Evidence Chains align with supply-chain frameworks such as SLSA.

\subsection{Summary}

The reference architecture decomposes software construction and operation into explicit roles: authors define admissibility, registries preserve protocols, generators search, validators establish compliance, runtime verifiers enforce deployed invariants, remediation components return failures to generation, and evidence stores preserve the basis for each admission and attestation decision.

\section{Runtime Evidence and Continuous Attestation}
\label{sec:runtime-evidence}

The Validator Loop establishes that a generated artifact is admissible at the time it enters the system. That claim is necessary but not sufficient for production governance. A deployed implementation may later encounter inputs, traffic patterns, dependency states, or concurrency schedules absent from pre-deployment validation. PDD therefore extends admission evidence into runtime evidence: the protocol remains the governing artifact, but evidence continues to accumulate after deployment.

Runtime evidence is necessarily scoped by observability. Some protocol obligations can be checked online from payloads, traces, resource counters, policy decisions, or state-transition events. Other obligations may require offline replay, sampling, formal analysis, or stronger instrumentation. Continuous attestation therefore covers the monitorable projection of the protocol rather than replacing build-time validation.

\subsection{The Static Evidence Boundary}

Build-time validation creates a controlled observation boundary. Validators can compile artifacts, execute property suites, inspect dependencies, run sandboxed workloads, and measure resource use under known conditions. These checks support disciplined admission, but they cannot exhaust the space of production executions. Live systems exhibit workload drift, partial failures, adversarial inputs, scheduler interleavings, dependency changes, and gradual resource degradation.

The limitation is not a defect in the Evidence Chain; it is a boundary condition. A development-time evidence object proves compliance relative to preserved validation inputs, validator assumptions, and the protocol observation model. It should not be read as a perpetual guarantee that every future execution will remain compliant.

\subsection{Dynamic Evidence Ledgers}

To represent compliance across the operational lifetime of a component, PDD generalizes the Evidence Chain into a \emph{Dynamic Evidence Ledger}. The initial ledger element records build-time admission:
\[
\mathcal{L}_0 = (E_{\mathrm{build}}).
\]
At runtime, each attestation interval produces a signed evidence block:
\[
E_t =
H(E_{t-1}, \mathcal{P}, I_v, R_t, A_t, t),
\]
where $E_{t-1}$ is the previous evidence block, $\mathcal{P}$ is the governing protocol, $I_v$ identifies the deployed implementation version, $R_t$ contains runtime observations over interval $t$, $A_t$ records the attestation decision, and $H$ denotes a cryptographic digest or signed attestation. The ledger evolves append-only:
\[
\mathcal{L}_t = \mathcal{L}_{t-1} \Vert E_t.
\]

The protocol itself remains the triplet
\[
\mathcal{P} = (\mathcal{S}, \mathcal{B}, \mathcal{O}).
\]
The deployed instance is time-indexed by implementation and evidence state:
\[
\mathcal{D}_t = (\mathcal{P}, I_v, \mathcal{L}_t).
\]
This distinction preserves the protocol as the durable specification while allowing operational evidence to evolve.

\begin{definition}[Monitorable Runtime Projection]
For protocol $\mathcal{P}$, let $\Omega_{\mathcal{P}}^r \subseteq \Omega_{\mathcal{P}}$ denote the protocol-relevant observations available to the runtime verifier. The runtime predicate $\Phi_{\mathcal{P}}^r$ is the restriction of $\Phi_{\mathcal{P}}$ to $\Omega_{\mathcal{P}}^r$.
\end{definition}

\begin{definition}[Continuous Protocol Attestation]
For protocol $\mathcal{P}$ and deployed implementation version $I_v$, continuous attestation over interval $t$ succeeds when every runtime observation in $R_t \subseteq \Omega_{\mathcal{P}}^r$ satisfies $\Phi_{\mathcal{P}}^r$, and the resulting decision is appended to $\mathcal{L}_t$ as signed evidence.
\end{definition}

Continuous attestation therefore changes the interpretation of compliance from a one-time admission event to a sequence of recorded claims over monitorable behavior. A component is not merely accepted once; it remains accountable to the protocol while it executes.

\subsection{Runtime Verification Layer}

Runtime evidence requires an enforcement boundary outside the generated implementation. We call this boundary the \emph{Runtime Verification Layer} (RVL). It can be realized through sidecars, service mesh policies, gateways, admission controllers, eBPF instrumentation, sandbox monitors, or other deployment mechanisms. The key property is isolation: the generated implementation may produce behavior, but it cannot unilaterally decide whether that behavior is compliant.

The RVL performs three functions:

\begin{enumerate}
    \item \textbf{Runtime observation:} collect protocol-relevant payloads, state transitions, resource measurements, dependency calls, and policy decisions.
    \item \textbf{Runtime enforcement:} block, quarantine, rate-limit, roll back, or degrade executions that violate structural, behavioral, or operational invariants.
    \item \textbf{Runtime attestation:} append signed evidence blocks to the Dynamic Evidence Ledger, including both successful attestations and violations.
\end{enumerate}

This layer is aligned with runtime verification: executions are monitored against formal or executable properties while the system runs~\cite{leucker2009runtime,havelund2004jpx}. PDD uses that pattern as part of a generated-software governance loop rather than as an isolated monitoring technique.

\begin{definition}[Runtime Verifier Soundness]
A runtime verifier for $\mathcal{P}$ is sound with respect to $\Phi_{\mathcal{P}}^r$ if every emitted pass block corresponds to observations satisfying $\Phi_{\mathcal{P}}^r$, and every emitted violation block identifies at least one observed $\omega \in \Omega_{\mathcal{P}}^r$ such that $\neg \Phi_{\mathcal{P}}^r(\omega)$.
\end{definition}

\begin{proposition}[Evidence-Carrying Runtime Violation]
Assume a sound runtime verifier for $\mathcal{P}$ and a violation block $E_t^{\mathrm{fail}}$ that cryptographically binds the deployed implementation version, runtime observation, protocol identity, verifier identity, and violated predicate. Then $E_t^{\mathrm{fail}}$ is evidence of non-compliance with the monitorable runtime projection of $\mathcal{P}$ for the recorded interval.
\end{proposition}

\begin{proof}[Proof]
By runtime verifier soundness, an emitted violation block identifies an observed $\omega \in \Omega_{\mathcal{P}}^r$ for which $\neg \Phi_{\mathcal{P}}^r(\omega)$. The cryptographic binding prevents the block from being reassigned to a different implementation, protocol, verifier, or observation without breaking the binding assumption. The block therefore records non-compliance with the monitorable projection of $\mathcal{P}$ for that interval.
\end{proof}

\subsection{Runtime Anomaly Taxonomy}

Runtime attestation is motivated by anomaly classes that are difficult to eliminate through pre-deployment validation alone:

\begin{itemize}
    \item \textbf{Structural drift:} live payloads, schema versions, or serialization behavior diverge from the typed handshake.
    \item \textbf{Behavioral drift:} observed outputs, error semantics, ordering guarantees, or monitorable idempotence properties diverge from $\mathcal{B}$.
    \item \textbf{Operational degradation:} latency, memory use, retries, cost, or dependency fan-out exceed $\mathcal{O}$ under production load.
    \item \textbf{Temporal and concurrency anomalies:} races, ordering violations, deadlocks, livelocks, or stale reads appear only under particular schedules.
    \item \textbf{Dependency and environment drift:} external services, package versions, configuration, or infrastructure state change after admission.
    \item \textbf{Authority and provenance violations:} an implementation attempts unapproved network access, file writes, privilege escalation, or unrecorded dependency use.
    \item \textbf{Distribution and adversarial shift:} production inputs exploit assumptions not represented in validation workloads.
\end{itemize}

These classes do not replace protocol authoring; they identify where protocols and validators may need refinement. A recurring anomaly can become a new structural, behavioral, or operational invariant in a later protocol version.

\subsection{Closed-Loop Remediation}

When runtime attestation fails, the RVL should not merely emit an alert. It should produce a structured invariant failure context that can be used by the generation and validation pipeline. Evidence blocks may store redacted observations, hashes, or pointers to access-controlled traces when raw telemetry contains sensitive data. Let
\[
E_t^{\mathrm{fail}} =
H(E_{t-1}, \mathcal{P}, I_v, R_t^{\mathrm{fail}}, A_t^{\mathrm{fail}}, t)
\]
denote a signed violation block. The remediation loop is:
\[
E_t^{\mathrm{fail}}
\rightarrow
C_t
\rightarrow
I'
\rightarrow
\mathrm{Validate}(I',\mathcal{P})
\rightarrow
\mathcal{L}_{t+1},
\]
where $C_t$ is a repair context derived from the violating observation, protocol clause, environment metadata, and ledger state. A generator may use $C_t$ to propose a patch $I'$, but the patch is not trusted because it was generated in response to a real failure. It must re-enter the Validator Loop and produce fresh admission evidence before deployment.

This closes the self-correction loop. Runtime failure becomes protocol-indexed evidence, not an informal bug report. The generator receives concrete proof of non-compliance, while the validator and RVL retain authority over admission and enforcement.

\subsection{Summary}

Dynamic Evidence Ledgers, Runtime Verification Layers, and closed-loop remediation extend PDD from static admission to continuous protocol governance. The protocol remains sovereign, the generator remains a proposal engine, and runtime execution becomes another source of signed evidence about whether observed deployed behavior remains within the monitorable protocol boundary.

\section{Case Studies}

We illustrate PDD with three examples: an idempotent handler, a bounded ETL pipeline, and an automatically generated microservice. These demonstrate how protocols govern implementations across scales.

\subsection{Idempotent User-Creation Handler}

Consider a service that creates a user account from a client-supplied identifier. In a conventional design document, the requirement might be written informally as:

\begin{quote}
``Create the user if it does not already exist, and return the existing record otherwise.''
\end{quote}

This statement leaves open required fields, duplicate recognition, admissible errors, external lookups, retries, and operational budgets.

Under PDD, these requirements are expressed as a protocol:

\begin{itemize}
    \item \textbf{Structural invariants $(\mathcal{S})$:}
    The request and response are defined by explicit schemas specifying required fields, optional metadata, and enumerated error variants.
    
    \item \textbf{Behavioral invariants $(\mathcal{B})$:}
    Repeated invocations with the same logical identifier must be idempotent:
    \[
    f(x,s)=(y,s') \Rightarrow f(x,s')=(y,s').
    \]
    Additional properties include deterministic errors and identifier uniqueness.
    
    \item \textbf{Operational invariants $(\mathcal{O})$:}
    At most one database write is permitted; no outbound network access is allowed; and end-to-end latency must remain below a specified bound.
\end{itemize}

An implementation generator is free to use optimistic insertion, uniqueness constraints, transactional lookup, or an equivalent strategy. The validation engine checks schema conformance, idempotence, error behavior, and operational limits; once admitted, the implementation can later be replaced without changing the protocol.

\subsection{Bounded ETL Pipeline}

As a second example, consider an extract-transform-load (ETL) component that normalizes transactional records before downstream analysis.

In a conventional specification, the requirement might be: ``process all valid records and reject malformed inputs.'' This leaves ordering, determinism, temporary disk usage, and memory budgets unspecified.

Under PDD, the pipeline is instead governed by a protocol that makes these constraints explicit:

\begin{itemize}
    \item \textbf{Structural invariants:}
    Input and output schemas define exact field types, nullability rules, and admissible record formats.
    
    \item \textbf{Behavioral invariants:}
    The transformation must be deterministic, schema-preserving where required, and conservative with respect to valid-record counts unless filtering or aggregation is explicitly permitted.
    
    \item \textbf{Operational invariants:}
    The pipeline must execute in streaming mode, remain within a fixed memory budget, avoid temporary disk writes, and satisfy a maximum per-record processing latency.
\end{itemize}

The implementation generator is free to realize the pipeline in Python, Rust, Apache Beam, or another framework. These choices are secondary; the generated implementation must satisfy the protocol's structural, behavioral, and operational constraints.

\subsection{Automated Microservice Generation}

The third case study illustrates the full Protocol-Driven Development lifecycle in an automated synthesis setting.

Suppose an architect introduces a fraud-detection microservice by defining a protocol bundle containing:

\begin{itemize}
    \item gRPC request and response schemas;
    \item behavioral properties such as deterministic scoring, monotonicity of risk under added evidence, and well-defined failure semantics;
    \item operational constraints such as latency budgets, approved feature stores, restricted outbound network access, and dependency on existing authentication and audit protocols.
\end{itemize}

An implementation generator retrieves the bundle and produces candidates using different prompts, libraries, or internal architectures. Each candidate is treated as a provisional proposal.

The validation engine checks interface conformance, determinism, monotonicity, regression properties, approved feature access, latency, and dependency constraints. Non-compliant candidates are rejected.

If a candidate satisfies the protocol bundle, it is admitted and recorded with its Discovery Log and signed evidence object. After deployment, a Runtime Verification Layer monitors feature-store calls, latency, response shape, and policy decisions. Suppose a production traffic shift causes the service to issue two feature-store calls for a subset of requests, violating the operational invariant that allows only one such call. The RVL records the violating trace as a signed ledger block and can rate-limit or quarantine the offending path. The remediation orchestrator converts the block into repair context for the generator; any proposed fix must pass the original Validator Loop before replacement.

Later, a new generator can also produce a more efficient realization; if the protocol is unchanged and the new artifact validates, the new version can replace the original without weakening the stated guarantees.

\subsection{Comparative Lessons}

Four properties recur: the protocol defines the stable component surface; regeneration and substitution become explicit operations; operational boundaries become first-class elements; and the same admission logic applies from functions to services.

\subsection{Implications}

The examples demonstrate incremental application of PDD, beginning with small module boundaries and extending to service architectures. PDD stabilizes what must remain stable: interface boundaries, behavioral properties, operational authority, and evidence of compliance.

\section{Evaluation Agenda}

PDD requires an empirical agenda focused on admission quality rather than code-generation speed alone. Implementations should be judged by whether protocol-governed construction yields artifacts that are less ambiguous, easier to regenerate, more substitutable, more operationally bounded, and more auditable than artifacts governed by prose and tests alone.

\subsection{Empirical Questions}

Future evaluations should organize around six questions:

\begin{enumerate}
    \item \textbf{Ambiguity Reduction:} Do protocols reduce implementation variance relative to narrative specifications and test suites alone?
    \item \textbf{Regeneration Reliability:} Can independently generated implementations repeatedly satisfy the same protocol across models, prompts, and languages?
    \item \textbf{Protocol-Level Substitutability:} Are distinct protocol-compliant implementations interchangeable for protocol-respecting clients?
    \item \textbf{Validation Overhead:} What computational and engineering costs are introduced by the Validator Loop?
    \item \textbf{Governance Efficacy:} Does protocol-based admission detect and block unauthorized side effects, capability violations, and provenance gaps?
    \item \textbf{Runtime Attestation:} Do runtime verifiers and Dynamic Evidence Ledgers detect post-deployment anomalies and produce repair evidence?
\end{enumerate}

These claims are empirically falsifiable. If protocols do not reduce protocol-visible variance, admitted implementations cannot be regenerated reliably, compliant implementations are not substitutable for protocol-respecting clients, the Validator Loop misses routine capability violations, or runtime attestation fails to detect monitorable post-deployment anomalies, then the central claims of PDD would require revision.

\subsection{Workloads and Baselines}

An empirical validation program should span multiple levels of system complexity:

\begin{itemize}
    \item \textbf{Functions:} idempotent handlers, parsers, validators, and deterministic transformations;
    \item \textbf{Data pipelines:} bounded ETL jobs with explicit memory, latency, and side-effect constraints;
    \item \textbf{Microservices:} gRPC and REST services with persistence, dependency policies, and runtime authority restrictions;
    \item \textbf{Automated regeneration tasks:} repeated implementation of the same protocol across multiple models, prompting strategies, and programming languages.
\end{itemize}

For each workload, future studies should compare at least three conditions:

\begin{enumerate}
    \item \textbf{Spec-Driven Development (SDD):} natural-language requirements and interface descriptions;
    \item \textbf{Test-Driven Development (TDD):} requirements supplemented by executable tests;
    \item \textbf{Protocol-Driven Development (PDD):} structural, behavioral, and operational invariants with validator-based admission.
\end{enumerate}

These baselines isolate the effect of treating the protocol bundle as the governing artifact for generated software.

\subsection{Ambiguity and the Natural Language Tax}

To measure ambiguity, a study should generate multiple implementations from the same design intent under SDD, TDD, and PDD conditions and compare them along three axes:

\begin{itemize}
    \item \textbf{Structural divergence:} differences in interface shape, field interpretation, schema compatibility, and version behavior;
    \item \textbf{Behavioral divergence:} differences in outputs, error semantics, determinism, and property satisfaction;
    \item \textbf{Operational divergence:} differences in external dependencies, side effects, latency, and resource usage.
\end{itemize}

These measurements make the \emph{Natural Language Tax} observable: interpretation cost appears as structural divergence, reconciliation cost as behavioral divergence, and maintenance cost as operational drift. Evidence would support PDD if protocol-visible divergence is lower than under SDD or TDD, and would weaken it if divergence remains comparable or greater despite the added protocol-authoring burden.

\subsection{Regeneration Reliability}

Future studies should test regeneration by repeatedly generating implementations for the same protocol using:

\begin{itemize}
    \item different generative models;
    \item multiple prompting strategies;
    \item varied sampling temperatures;
    \item alternative programming languages and runtime frameworks.
\end{itemize}

The relevant metrics are:

\begin{itemize}
    \item \textbf{validation pass rate,}
    \item \textbf{number of attempts to first successful admission,}
    \item \textbf{time to successful admission,}
    \item \textbf{rate of protocol-level substitutability among admitted implementations.}
\end{itemize}

Evidence would support PDD if the protocol remains a stable target across generation methods and admitted implementations remain substitutable under the protocol's observation model. It would weaken PDD if small changes in model, prompt, or language commonly produce candidates that cannot satisfy the protocol or require extensive manual repair.

\subsection{Protocol-Level Substitutability}

Substitutability should be evaluated by placing independently generated implementations of the same protocol behind the same interface boundary and measuring:

\begin{itemize}
    \item functional outputs and error semantics;
    \item client-visible timing behavior within admissible operational envelopes;
    \item downstream compatibility with dependent services;
    \item protocol-visible side effects and authority usage.
\end{itemize}

This evaluation asks whether a protocol-respecting client can replace one admitted implementation with another without relying on behavior outside the protocol. Failures indicate underspecified protocols, unsound validators, or clients coupled to implementation accidents.

\subsection{Validation Cost}

The Validator Loop adds work relative to conventional build-and-test pipelines. An empirical assessment should measure:

\begin{itemize}
    \item structural validation time;
    \item behavioral validation time, including property-based and regression checks;
    \item operational validation overhead introduced by sandboxing and instrumentation;
    \item evidence generation and recording latency;
    \item runtime attestation overhead for live monitoring;
    \item total time from candidate generation to admission decision.
\end{itemize}

The central question is whether added cost is proportionate to gains in governance, substitutability, and auditability. A negative result occurs if validation overhead dominates without improving admission quality.

\subsection{Governance Efficacy}

To evaluate governance directly, future studies should include non-compliant candidates that violate operational or provenance constraints:

\begin{itemize}
    \item unauthorized network access;
    \item hidden temporary file writes;
    \item excessive database calls;
    \item use of unapproved dependencies;
    \item latency-budget violations;
    \item missing or malformed provenance metadata;
    \item production-only races, memory leaks, dependency drift, and workload spikes.
\end{itemize}

The main metric is the fraction of violations detected and blocked by the Validator Loop before admission and by the Runtime Verification Layer after deployment. PDD is strengthened if operational and provenance violations are systematically rejected before admission and post-deployment anomalies produce signed evidence and targeted repair contexts. It is weakened if non-compliant candidates pass validation at rates comparable to conventional test-only pipelines or if deployed violations are visible in telemetry but absent from the Dynamic Evidence Ledger.

\subsection{Evidence and Reproducibility}

For every admitted implementation, the Discovery Log, Evidence Chain, and Dynamic Evidence Ledger should be assessed by whether they allow an auditor or downstream system to:

\begin{itemize}
    \item reconstruct the exact protocol version and dependency closure;
    \item verify artifact hashes and provenance metadata;
    \item replay validation decisions deterministically when the required inputs are preserved;
    \item trace the admission decision from protocol to implementation to deployment artifact;
    \item inspect runtime attestation blocks, violation evidence, and remediation outcomes.
\end{itemize}

This dimension asks whether PDD produces usable admission and runtime evidence rather than only validation reports. Strong evidence would show decisions and runtime observations that are inspectable and linked to the governing protocol; weak evidence would show an evidence record too incomplete, expensive, or fragile to support audit.

\subsection{Validation and Falsification Criteria}

The empirical agenda can be summarized as a set of falsifiable expectations:

\begin{enumerate}
    \item PDD should reduce protocol-visible ambiguity relative to SDD and TDD baselines.
    \item Regenerated implementations should repeatedly reach admission under the same protocol.
    \item Independently generated implementations should be substitutable for protocol-respecting clients.
    \item Validation overhead should be measurable and justified by governance benefits.
    \item Operational and provenance violations should be detected before admission.
    \item Runtime verification should detect anomalies that appear after deployment and produce signed evidence blocks.
    \item Evidence artifacts should support reproducible admission decisions, runtime attestation, and retrospective audit.
\end{enumerate}

Failure on any dimension would be informative: the protocol may be too weak, the validator unsound, the evidence insufficient, or authoring costs too high for the target system class.

\subsection{Summary}

This agenda specifies how PDD should be tested once implemented in concrete toolchains. To serve as authoritative development artifacts, protocols must measurably stabilize generated software, bound operational authority, and produce evidence for admission, runtime attestation, regeneration, substitution, remediation, and audit.

\section{Related Work}

Protocol-Driven Development (PDD) draws on formal verification, runtime verification, executable testing, declarative infrastructure, software supply-chain security, and automated software engineering. Its contribution is not a standalone verification primitive, but a development model in which protocols are the enduring artifact and implementations are admitted and continuously attested through validation evidence.

Table~\ref{tab:pdd-adjacent-comparison} summarizes this positioning. The adjacent approaches are complementary; PDD differs by making a protocol bundle the primary artifact that jointly governs structure, behavior, operational authority, evidence-producing admission, and runtime attestation.

\begin{table*}[t]
\centering
\caption{Protocol-Driven Development compared with adjacent approaches.}
\label{tab:pdd-adjacent-comparison}
\footnotesize
\setlength{\tabcolsep}{3pt}
\renewcommand{\arraystretch}{1.08}
\begin{tabular}{@{}>{\raggedright\arraybackslash}p{0.21\textwidth}>{\raggedright\arraybackslash}p{0.13\textwidth}>{\raggedright\arraybackslash}p{0.13\textwidth}>{\raggedright\arraybackslash}p{0.15\textwidth}>{\raggedright\arraybackslash}p{0.15\textwidth}>{\raggedright\arraybackslash}p{0.15\textwidth}@{}}
\toprule
\textbf{Approach} & \textbf{Structure} & \textbf{Behavior} & \textbf{Operational authority} & \textbf{Evidence / provenance} & \textbf{Primary artifact} \\
\midrule
Interface schemas / API descriptions & Schemas, signatures & Mostly external & Mostly external & Version metadata & Interface contract \\
TDD / property-based testing & Indirect via fixtures & Examples, properties & Usually external & Test results & Test suite \\
Policy-as-code / sandbox policies & Policy inputs & Policy predicates & Capabilities, resources & Decision logs & Policy rule set \\
Supply-chain provenance / attestations & Artifact metadata & Not specified & Not specified & Attestations, lineage & Provenance record \\
Runtime verification & Instrumented events & Runtime properties & Monitor actions & Execution traces & Monitor specification \\
\textbf{Protocol-Driven Development (PDD)} & \textbf{Typed handshakes} & \textbf{Behavioral invariants} & \textbf{Capability manifest} & \textbf{Evidence Ledger} & \textbf{Protocol bundle} \\
\bottomrule
\end{tabular}
\end{table*}

The distinction is artifact scope rather than replacement. PDD uses schemas, tests, policies, and provenance mechanisms as components of protocol-governed admission, while treating the protocol bundle as the durable admission artifact.

\subsection{Formal Verification, Testing, and Executable Correctness}

The formal foundations of PDD lie in formal methods. Hoare logic, model checking, TLA+, and Alloy show that correctness can be expressed through machine-checkable assertions, state-space exploration, invariants, temporal properties, and relational constraints~\cite{hoare1969axiomatic,clarke1999model,lamport2002specifying,jackson2002alloy}. These traditions support a premise of PDD: some correctness boundaries are more stable than the code that realizes them.

Test-Driven Development made executable artifacts central to engineering practice~\cite{beck2003tdd}; property-based testing generalized examples into laws over generated inputs~\cite{claessen2000quickcheck}.

\subsection{From Type Safety to Protocol Admissibility}

Type systems, dependent types, and Proof-Carrying Code (PCC) establish the foundation for statically verifying that programs meet semantic and safety constraints before execution~\cite{wright1994syntactic,swamy2016dependent,appel2001foundational}. Recent LLM-assisted synthesis frameworks adopt these formalisms to guide code generation against strict correctness obligations~\cite{cai2025automated,mukherjee2025synver}. 

PDD generalizes this intuition from program correctness to operational governance. While traditional type-checking and PCC are largely confined to language boundaries or safety policies, PDD treats generated implementations as untrusted candidates that must satisfy broader admissibility criteria. By combining static verification with runtime contract enforcement, PDD ensures that an artifact achieves structural validity, behavioral conformance, operational authorization, and evidentiary accountability before it is permitted to enter a governed system.

\subsection{Runtime Verification and Online Monitoring}

Runtime verification monitors executions against formal or executable properties after deployment~\cite{leucker2009runtime,havelund2004jpx}. PDD adopts this idea for generated software governance but changes the artifact boundary: the runtime monitor is not a standalone specification divorced from development. It evaluates observations against the same protocol bundle that governed admission, appends signed results to a Dynamic Evidence Ledger, and can feed violation evidence back into the generation and validation loop.

\subsection{Specifications, Interfaces, and Declarative Artifacts}

Traditional software engineering has long relied on specifications ranging from prose requirements to machine-readable interfaces. Systems such as OpenAPI, JSON Schema, and Protocol Buffers stabilize communication boundaries and enable structural automation~\cite{openapi2021,jsonschema2022,protobuf2008}, but they typically capture shape and compatibility rather than full behavioral and operational semantics.

\subsection{Policy-as-Code, Zero Trust, and Supply-Chain Provenance}

The paper also draws on policy-oriented governance. Zero-trust architecture replaces implicit trust with explicit verification~\cite{rose2020zerotrust}, and policy-as-code systems such as Open Policy Agent decouple authorization logic from application code~\cite{opa2016}. These approaches show the value of declarative governance evaluated independently of the governed artifact.

Supply-chain frameworks such as in-toto and SLSA extend this principle to provenance, grounding trust in attestations, traceability, and build metadata~\cite{torresarias2019intoto,slsa2021}.

\subsection{Automated Software Engineering and Program Synthesis}

The motivating context for PDD is the rise of code-generating language models and automated synthesis engines. The ``Software 2.0'' framing recast parts of software construction as search over program space~\cite{karpathy2017software}; recent systems synthesize programs from prompts, improve productivity on selected tasks, and tackle repository-level software engineering benchmarks~\cite{chen2021codex,peng2023copilot,yang2024sweagent,jimenez2024swebench,wu2024devin}.

Much of this literature focuses on generation capability, repair ability, and task completion. Here the emphasis is admission and governance.

\subsection{Relationship to Runtime Governance}

PDD is related to runtime governance for probabilistic systems, where model output is treated as a proposal, authority is placed in explicit constraints, and consequential admissions preserve evidence. Sovereign Agentic Loops and OpenKedge are runtime examples: SAL separates reasoning from execution through a control boundary~\cite{he2026sal}, while OpenKedge evaluates runtime operations through policy, contracts, and evidence records~\cite{he2026openkedge}.

PDD is the development-time counterpart to this pattern and extends it with continuous attestation. Runtime governance asks whether an action may execute; PDD asks whether a generated implementation is admissible before it becomes part of the system and whether its deployed behavior remains within protocol bounds. In both cases, a proposal is evaluated against explicit constraints and acceptance is accompanied by evidence.

This relationship is contextual rather than necessary. PDD stands on its own protocol model: structural invariants, behavioral invariants, operational authority, the Validator Loop, and Evidence Chains. Testing asks whether code behaves as expected in selected cases; PDD asks whether generated code has crossed a development boundary in which typed handshakes, behavioral laws, operational capabilities, validation evidence, and provenance align.

\subsection{Synthesis: Positioning of PDD}

Prior work provides mature mechanisms for interface definition, formal specification, policy enforcement, provenance, and code generation, but these mechanisms usually operate in separate lifecycle stages. PDD combines them into a development-time admission model: it generalizes interfaces from structural shape to admissibility boundaries, incorporates evidence into admission, and treats generated implementations as candidates rather than authorities.

\section{Discussion and Future Work}

Protocol-Driven Development (PDD) reframes software engineering around a shift in durable artifacts: implementation code is expected to change frequently, while protocols define admissible behavior. We consider implications, limitations, and directions for future research.

\subsection{From Programming to Protocol Engineering}

Under PDD, design effort shifts from writing a single implementation to authoring protocols that admit many possible implementations. Protocol authoring becomes a primary engineering activity: architects define typed handshakes, behavioral assertions, capability manifests, and validator requirements. The implementation becomes a replaceable witness of these constraints.

\subsection{Protocol Reuse and Protocol Registries}

A natural next step is reusable protocol registries. Just as package repositories support library reuse, protocol registries would support reuse of interface contracts, behavioral constraints, and capability boundaries for recurring patterns such as idempotent APIs, financial handlers, or audit-logging services.

\subsection{Regenerable Systems}

Repeated regeneration is possible in response to improved generators, dependency vulnerabilities, language migration, or runtime violations. The admission condition remains unchanged: regenerated implementations must satisfy the governing protocol and produce valid evidence before replacing earlier realizations. Runtime evidence can make regeneration more targeted by identifying the violated invariant, input class, dependency state, or resource envelope that triggered repair.

\subsection{Continuous Attestation}

Static evidence should not be treated as permanent proof of production behavior. Dynamic Evidence Ledgers extend auditability across the operational lifetime of a component by recording signed runtime observations, invariant checks, violations, and remediation outcomes. This shifts PDD from a deployment gate to a continuous governance loop: admission remains necessary, but deployed implementations continue to be measured against the monitorable projection of the protocol.

\subsection{Protocol Inference and Synthesis}

We assume protocols are authored by architects or human-machine collaborations. Future systems could infer invariants from existing code, telemetry, Discovery Logs, or regulatory documents, but machine-synthesized protocols raise open questions about validator trust, protocol quality, and verification of generated constraints.

\subsection{Integration with Stronger Formal Methods}

PDD is compatible with verification techniques ranging from property testing to formal proof. The Validator Loop can incorporate theorem provers, SMT solvers, model checkers, and certified compilers as mechanisms for making formal verification an operational component of software admission.

\subsection{Economic Implications}

Automated program synthesis makes candidate implementations cheaper while trustworthy specification and governance remain difficult. Under PDD, durable engineering assets include protocol libraries, validator implementations, and domain-specific capability policies, not proprietary source code alone.

\subsection{Limitations}

Protocol-Driven Development does not eliminate engineering failure. A poorly authored protocol may miss intent; an unsound validator may admit non-compliant implementations; runtime monitors may introduce overhead, false positives, blind spots, or privacy constraints; and some properties may be expensive, incomplete, or undecidable to verify. Runtime evidence can also be misleading if telemetry is sampled poorly, if sensitive observations are over-redacted, or if the monitorable projection omits the property that actually failed. Protocol authoring may be unjustified for very small, short-lived, or low-risk systems. PDD shifts difficulty from implementation authorship toward protocol quality, validator soundness, runtime monitor fidelity, and evidence integrity.

\subsection{Relationship to Sovereign Engineering}

Beyond the development model proposed here, PDD fits a broader vision of \emph{sovereign engineering}: consequential transitions are mediated by explicit constraints and verifiable evidence. In this setting, decision authority belongs to a governing artifact or control boundary rather than to the probabilistic generator that produced a proposal.

OpenKedge and Sovereign Agentic Loops provide runtime examples of this pattern. PDD applies the same idea to software construction, governing whether proposed implementations may enter the codebase and whether deployed implementations remain within their protocol envelope. The framework remains self-contained as a protocol model while allowing runtime systems to contribute continuous evidence.

\subsection{Summary}

Protocol-Driven Development defines a design space in which software systems are specified through protocols, regenerated under explicit constraints, and governed through validation and runtime evidence. Future work includes protocol registries, automated protocol synthesis, formal-verification integration, runtime attestation policies, and deployment studies.

\section{Conclusion}

Automated synthesis engines have changed the economics of software engineering. Candidate implementations are becoming cheaper to produce and easier to regenerate. Under these conditions, the primary challenge is no longer only to author code efficiently, but to define the boundaries within which generated code may be admitted, composed, and replaced.

We proposed \textbf{Protocol-Driven Development (PDD)}, a software engineering model in which the primary artifact is a machine-enforceable protocol rather than implementation code. We defined a protocol as a triplet of structural, behavioral, and operational invariants that together characterize the admissible implementation space of a software component. Under this formulation, implementations are transient realizations discovered through constrained search rather than permanent objects of engineering preservation.

We outlined the \emph{Validator Loop} to separate protocol authoring, candidate generation, verification, and admission into distinct roles. Automated generators may explore the implementation space, but an artifact becomes admissible only when a validation engine establishes compliance and emits verifiable evidence. Discovery Logs and Evidence Chains make acceptance auditable and accountable; replay is supported when validator inputs are preserved.

We then extended this admission model into operation through Dynamic Evidence Ledgers and a Runtime Verification Layer. Runtime observations, invariant checks, violation reports, and remediation outcomes can be appended as signed evidence, allowing deployed implementations to remain accountable to the protocol after release. When runtime attestation fails, the violation becomes structured repair context, but any generated patch must still pass validation before replacement.

We further developed a theoretical foundation for PDD by modeling software construction as search over a protocol-defined implementation set and formalizing basic properties such as sound acceptance, protocol-level substitutability, and monotonic strengthening under protocol refinement. These results formalize the intuition that implementations satisfying the same protocol can vary internally while remaining interchangeable under the protocol's observation model.

We also presented a reference architecture, illustrative case studies, and an evaluation agenda outlining how PDD can be applied and tested across functions, pipelines, and automated microservices. By combining ideas from formal methods, executable testing, runtime verification, policy-as-code, software provenance, and automated software engineering, PDD defines a governance model for software synthesis under low-cost implementation generation.

As implementation cost approaches zero, the enduring intellectual product of software engineering becomes the set of formal constraints that define admissible behavior, together with the evidence that those constraints were respected.

Protocol-Driven Development frames software construction and operation as constrained discovery with explicit evidence. It points toward systems in which reasoning, implementation, and execution are treated as proposals admitted through formal boundaries and verifiable evidence. Code remains changeable, but the protocol carries the durable authority to define which generated software artifacts are admissible and whether deployed behavior remains within bounds.

\appendix
\section{Appendix: Illustrative Minimal PDD Artifact Bundle}
\label{app:minimal-pdd-bundle}

This appendix sketches an illustrative minimal PDD artifact bundle. The example is not prescriptive: concrete systems can use OpenAPI, Protocol Buffers, JSON Schema, Rego, TLA+, property-based testing frameworks, in-toto attestations, SLSA metadata, or other formats. The artifact blocks are schematic rather than implementation syntax. The protocol bundle collects structural, behavioral, and operational constraints into one versioned artifact, and the evidence record binds an admitted implementation back to that artifact.

\newenvironment{artifactblock}
{\par\smallskip\noindent\small\renewcommand{\arraystretch}{0.95}\begin{tabular}{@{}>{\raggedright\arraybackslash}p{0.98\columnwidth}@{}}}
{\end{tabular}\par\smallskip}

\subsection{Bundle Structure}

A minimal bundle contains a manifest, invariant files, validator requirements, and an evidence namespace. One directory layout is:

\begin{artifactblock}
fraud-score.protocol/\\
\hspace*{1em}protocol.yaml\\
\hspace*{1em}structural/\\
\hspace*{2em}request-response.schema.yaml\\
\hspace*{1em}behavioral/\\
\hspace*{2em}scoring.properties.yaml\\
\hspace*{1em}operational/\\
\hspace*{2em}capabilities.yaml\\
\hspace*{1em}validators/\\
\hspace*{2em}validator-set.yaml\\
\hspace*{1em}evidence/\\
\hspace*{2em}admission-record.json\\
\hspace*{2em}runtime-ledger.jsonl
\end{artifactblock}

The top-level manifest identifies the protocol, its version, and the invariant artifacts that define admission:

\begin{artifactblock}
protocol\_id: fraud-score\\
version: 1.0.0\\
component: risk.scoring.FraudScore\\
invariants:\\
\hspace*{1em}structural: structural/request-response.schema.yaml\\
\hspace*{1em}behavioral: behavioral/scoring.properties.yaml\\
\hspace*{1em}operational: operational/capabilities.yaml\\
validators:\\
\hspace*{1em}required\_set: validators/validator-set.yaml\\
evidence:\\
\hspace*{1em}namespace: evidence/
\end{artifactblock}

\subsection{Structural Invariant}

The structural invariant fixes the typed handshake. A concrete implementation could encode it in JSON Schema, OpenAPI, Protocol Buffers, or another schema language.

\begin{artifactblock}
request:\\
\hspace*{1em}type: object\\
\hspace*{1em}required: [transaction\_id, account\_id, amount\_cents]\\
\hspace*{1em}properties:\\
\hspace*{2em}transaction\_id: string\\
\hspace*{2em}account\_id: string\\
\hspace*{2em}amount\_cents: integer, minimum 0\\
\hspace*{2em}merchant\_country: ISO-3166 alpha-2 string\\
response:\\
\hspace*{1em}type: object\\
\hspace*{1em}required: [transaction\_id, risk\_score, decision]\\
\hspace*{1em}properties:\\
\hspace*{2em}transaction\_id: string\\
\hspace*{2em}risk\_score: number, range [0.0, 1.0]\\
\hspace*{2em}decision: approve | review | decline\\
errors: invalid\_request | dependency\_unavailable
\end{artifactblock}

\subsection{Behavioral Invariant}

The behavioral invariant records protocol-visible semantic properties. The pseudo-format below is intentionally neutral; a concrete validator could translate these entries into property-based tests, metamorphic checks, or formal assertions.

\begin{artifactblock}
properties:\\
\hspace*{1em}- name: deterministic\_scoring\\
\hspace*{2em}for\_all: request\\
\hspace*{2em}require: score(request) == score(request)\\
\hspace*{1em}- name: score\_range\\
\hspace*{2em}for\_all: request\\
\hspace*{2em}require: 0.0 \(\leq\) risk\_score \(\leq\) 1.0\\
\hspace*{1em}- name: monotone\_amount\_risk\\
\hspace*{2em}for\_all: request\_a, request\_b\\
\hspace*{2em}when: same fields except amount\_cents\\
\hspace*{2em}require: larger amount does not lower risk\_score\\
\hspace*{1em}- name: invalid\_request\_fails\_closed\\
\hspace*{2em}when: missing required field\\
\hspace*{2em}require: error.kind == invalid\_request
\end{artifactblock}

\subsection{Operational Invariant}

The operational invariant (often expressed as a capability manifest) constrains what the generated implementation is permitted to do while satisfying the structural and behavioral contract.

\begin{artifactblock}
capabilities:\\
\hspace*{1em}network:\\
\hspace*{2em}outbound\_allowlist: [feature-store.internal:443]\\
\hspace*{2em}deny\_other\_outbound: true\\
\hspace*{1em}filesystem:\\
\hspace*{2em}read: []\\
\hspace*{2em}write: []\\
\hspace*{1em}dependencies:\\
\hspace*{2em}allow: [risk-common, protocol-runtime]\\
\hspace*{1em}resources:\\
\hspace*{2em}max\_latency\_ms\_p95: 75\\
\hspace*{2em}max\_memory\_mb: 256\\
\hspace*{2em}max\_feature\_store\_calls\_per\_request: 1\\
\hspace*{1em}secrets:\\
\hspace*{2em}allow: [FEATURE\_STORE\_TOKEN]\\
\hspace*{1em}background\_work:\\
\hspace*{2em}allowed: false
\end{artifactblock}

\subsection{Evidence Object}

An evidence object records the admission decision for one candidate implementation. The record is not the protocol itself; it is a signed link between a protocol version, a candidate artifact, validator execution, and observed validation results.

\begin{artifactblock}
evidence\_id: evd\_2026\_05\_11\_001\\
protocol:\\
\hspace*{1em}protocol\_id: fraud-score\\
\hspace*{1em}version: 1.0.0\\
\hspace*{1em}bundle\_digest: sha256:3b8f...\\
implementation:\\
\hspace*{1em}artifact\_id: risk-scoring-service\\
\hspace*{1em}artifact\_digest: sha256:91ac...\\
\hspace*{1em}language: python\\
\hspace*{1em}runtime: python-3.12\\
validators:\\
\hspace*{1em}- schema-conformance, version 0.4.2, result pass\\
\hspace*{1em}- property-check, version 0.9.1, result pass\\
\hspace*{2em}generated\_cases: 5000\\
\hspace*{2em}counterexamples: 0\\
\hspace*{1em}- capability-monitor, version 0.3.0, result pass\\
\hspace*{2em}max\_latency\_ms\_p95: 61\\
\hspace*{2em}network\_violations: 0\\
\hspace*{2em}filesystem\_writes: 0\\
decision: admit\\
issued\_at: 2026-05-11T00:00:00Z\\
issuer: validation-engine.example\\
signature: sig:...
\end{artifactblock}

\subsection{Runtime Ledger Block}

After deployment, runtime attestation blocks can be appended to the evidence namespace. The example below records a monitorable operational violation without requiring the ledger to expose raw customer data.

\begin{artifactblock}
ledger\_block\_id: evd\_2026\_05\_11\_runtime\_0042\\
previous\_block\_digest: sha256:aa17...\\
protocol:\\
\hspace*{1em}protocol\_id: fraud-score\\
\hspace*{1em}version: 1.0.0\\
implementation:\\
\hspace*{1em}artifact\_digest: sha256:91ac...\\
\hspace*{1em}deployed\_version: risk-scoring-service@2026.05.11\\
runtime\_verifier:\\
\hspace*{1em}name: capability-monitor\\
\hspace*{1em}version: 0.3.0\\
interval:\\
\hspace*{1em}start: 2026-05-11T00:05:00Z\\
\hspace*{1em}end: 2026-05-11T00:06:00Z\\
attestation:\\
\hspace*{1em}decision: violation\\
\hspace*{1em}violated\_invariant: max\_feature\_store\_calls\_per\_request\\
\hspace*{1em}observed\_value: 2\\
\hspace*{1em}allowed\_value: 1\\
\hspace*{1em}trace\_digest: sha256:f03c...\\
\hspace*{1em}raw\_trace\_location: access-controlled://traces/f03c...\\
action:\\
\hspace*{1em}runtime: route quarantined\\
\hspace*{1em}remediation\_context: repairctx\_2026\_05\_11\_0042\\
signature: sig:...
\end{artifactblock}

The artifacts illustrate the separation central to PDD. The protocol bundle defines the admissibility boundary, the implementation is a candidate realization, the admission evidence records why that candidate entered the system, and runtime ledger blocks record whether deployed behavior remains within the monitorable protocol boundary. Other systems can choose different concrete encodings while preserving the same artifact roles.

\bibliography{references}

\end{document}